\documentclass{amsart}
\usepackage{geometry}
\usepackage{cite}

\newtheorem{thm}{Theorem}[section]
\newtheorem{cor}{Corollary}[section]
\newtheorem{conj}{Conjecture}[section]
\newtheorem{prop}{Proposition}[section]
\newtheorem{lem}{Lemma}[section]
\newtheorem{remark}{Remark}[section]

\theoremstyle{definition}

\newtheorem{definition}{Definition}[section]

\newcommand{\Nset}{\mathbb{N}}

\newcommand{\Rset}{\mathbb{R}}

\newcommand{\bphi}{\bar{\varphi}}
\newcommand{\Wl}{\varphi_\lambda}
\newcommand{\Wlj}{\varphi_{\lambda,j}}
\newcommand{\Wlk}{\varphi_{\lambda,k}}
\newcommand{\Wljk}{\varphi_{\lambda,j,k}}
\newcommand{\dl}{d_\lambda}

\newcommand{\Wr}{\operatorname{Wr}}
\newcommand{\ord}{\operatorname{ord}}

\usepackage{enumerate}
\usepackage{multirow}
\begin{document}

\title[]{Oscillation theorems for the Wronskian of an arbitrary sequence of eigenfunctions of Schr\"odinger's equation}

\author{M\textordfeminine  \'Angeles Garc\'ia-Ferrero}
\address{Departamento de F\'isica Te\'orica II, Universidad Complutense de
Madrid, 28040 Madrid, Spain.}

\author{David G\'omez-Ullate}
\thanks{This work has been partially supported by the Spanish MINECO-FEDER Grants MTM2012-31714 and FIS2012-38949-C03-01.} 
\address{Instituto de Ciencias Matem\'aticas (CSIC-UAM-UC3M-UCM),  C/ Nicolas Cabrera 15, 28049 Madrid, Spain.}
\address{Departamento de F\'isica Te\'orica II, Universidad Complutense de
Madrid, 28040 Madrid, Spain.}
\email{mariangelesgferrero@gmail.com, david.gomez-ullate@icmat.es }

\begin{abstract}
The work of Adler provides necessary and sufficient conditions for the Wronskian of a given sequence of eigenfunctions of Schr\"odinger's equation to have constant sign in its domain of definition. We extend this result by giving explicit formulas for the number of real zeros of the Wronskian of an arbitrary sequence of eigenfunctions.  Our results apply in particular to Wronskians of classical orthogonal polynomials, thus generalizing classical results by Karlin and Szeg\H{o}.
Our formulas hold under very mild conditions that are believed to hold for generic values of the parameters. In the Hermite case, our results allow to prove some conjectures recently formulated by Felder et al.
\end{abstract}
\maketitle

\section{Introduction and main results}

Consider the Schr\"odinger eigenvalue problem $H[\varphi]=E\varphi$ where the Hamiltonian
\begin{eqnarray}\label{eq:H}
H[\varphi] = -\varphi'' + V(x) \varphi,  \qquad x\in(a,b) 
\end{eqnarray}
is assumed to have a pure-point spectrum given by square integrable eigenfunctions $\{\varphi_n\}_{n=0}^\infty$ with eigenvalues $E_0<E_1<E_2<\cdots$.
We assume that the potential $V(x)$ is regular in $(a,b)$ and that 
\begin{equation}
\varphi_n(a)=\varphi_n(b)=0, \quad n\in\Nset
\end{equation}
where the equalities have to be interpreted in the limit sense if the endpoints $a$ or $b$ are infinity.

From standard oscillation theorems, we know that $\varphi_n$ has $n$ simple zeros in $(a,b)$ and that the zeros of two consecutive eigenfunctions interlace.
The purpose of this paper is to derive oscillation theorems for the Wronskian determinant $\Wr[\varphi_{k_1},\dots,\varphi_{k_\ell}]$ of an arbitrary sequence of eigenfunctions. More precisely, to find out how many real roots it has in $(a,b)$. 
The main interest in this question stems from the theory of Darboux transformations, which are used in the dressing method to generate new solutions to an integrable system from known ones \cite{Freeman1983}, or with a similar scope in the factorization method in quantum mechanics, \cite{Infeld}. Crum \cite{Crum1955} showed that higher order or iterated Darboux transformations with seed functions $\varphi_{k_1},\dots,\varphi_{k_\ell}$ of a given potential $V(x)$ result in a transformed potential $\tilde V$ given by
\[\tilde V= V-2 D_{xx} \log \Wr[\varphi_{k_1},\dots,\varphi_{k_\ell}]. \]
The natural question to avoid singularities in the transformed potential is to characterize which sequences of eigenfunctions are such that their Wronskian determinant does not vanish in $(a,b)$. A necessary and sufficient condition on the sequence was given by Adler in \cite{Adler}:
\begin{thm}[Adler]\label{thm:Adler}
The Wronskian determinant of a sequence of eigenfunctions $\varphi_{k_1},\dots,\varphi_{k_\ell}$ of \eqref{eq:H} has constant sign in $(a,b)$ if and only if the sequence $(k_1,\dots,k_\ell)$ is built by concatenation of the following subsequences
\begin{enumerate}
\item[i)] a segment of consecutive integers of arbitrary length starting at $0$
\item[ii)] any number of segments of consecutive integers of even length.
\end{enumerate}
\end{thm}
An alternative characterization of those sequences (albeit only proving sufficiency) for which the Wronskian has no zeros was given by Krein in \cite{Krein1957} as the sequences for which the polynomial 
\begin{equation}\label{eq:kreincond}
 p(x)=(x-k_1)(x-k_2)\cdots(x-k_\ell),\text{ satisfies } p(n)\geq 0, \text{ for all } n\in\Nset.
 \end{equation}
The results of Krein and Adler  have been recently extended to multiple Darboux transformations of mixed type, \cite{GGM13}.

Wronskian determinants of orthogonal polynomials have been studied since the early works of Karlin and Szeg\H{o} \cite{Karlin1960}. They proved the following theorem concerning the Wronskian of a sequence of consecutive orthogonal polynomials
\begin{thm}[Karlin and Szeg\H{o}]\label{thm:Karlin}
Let $\{P_n\}_{n=0}^\infty$ be orthogonal polynomials with respect to an arbitrary measure whose distribution function has an infinite number of increasing points. Then the Wronskian determinant
\begin{equation}
W(n,\ell,x)=\Wr[P_n,P_{n+1},\dots,P_{n+\ell-1}]
\end{equation}
has constant sign if $\ell$ is odd, and  $n$ simple real zeros in the support of the measure if $\ell$ is even. Moreover, in the last case the real zeros of $W(n,\ell,x)$ and $W(n+1,\ell,x)$ strictly interlace.
\end{thm}

For certain sequences of  classical orthogonal polynomials, their Wronskian determinant defines a family of exceptional orthogonal polynomials \cite{Gomez-Ullate2009a,Gomez-Ullate2010a,MR2727790,Sasaki1,MR2845760,GKMFound}, a complete family of Sturm-Liouville orthogonal polynomials where some degrees are missing. The orthogonality weight for such families is the classical weight divided by the square of a Wronskian of classical polynomials, so the question of when such objects have constant sign becomes essential to ensure a well defined orthogonal polynomial system. A renewed interest in this matter comes from the recent discovery of Dur\'an \cite{Duran2014a,Duran2014b} that Christoffel transformations of classical discrete measures for orthogonal polynomials lead to discrete Krall polynomials which in turn are related by duality to exceptional discrete orthogonal polynomials. The positivity of the perturbed measure, which imposes conditions \eqref{eq:kreincond} in the Charlier case \cite{Duran2014a}, and similar conditions in the Meixner case \cite{Duran2014b}, translates into a well defined weight for the exceptional polynomials under the usual limit procedure. For exceptional Hermite polynomials the classification is complete \cite{gomez2013rational}, and every such polynomial can be expressed as a Wronskian determinant of Hermite polynomials such as those studied in \S \ref{sec:Hermite} .
Exceptional orthogonal polynomials have \textit{regular} zeros (which lie in the support of the measure) and \textit{exceptional} ones (which lie outside the support of the measure). Some interlacing and asymptotic properties of such zeros are given in \cite{Gomez-Ullate2013a}.
Certain Wronskians of Hermite polynomials allow to build rational solutions to nonlinear differential equations such as PIV and NLS, \cite{clarkson1,clarkson2,FP2008}.
Their complex roots form very regular patterns in the complex plane, \cite{clarkson1}, which can be interpreted approximately in terms of the Ferrer's diagram of the partition that defines the sequence, \cite{Felder2012a}. Zhang and Filipuk have recently studied Wronskian determinants of multiple orthogonal polynomials, \cite{Zhang2014}.

The results of this paper are a natural generalization of Theorem \ref{thm:Adler}. Under rather mild non-degeneracy conditions, we derive a formula for the number of real zeros of the Wronskian of an arbitrary sequence of eigenfunctions of \eqref{eq:H}. In case the potential in $\eqref{eq:H}$ is even,  the symmetry properties of the Wronskian entail a modification to the previous formula.

Given an indexed family of functions $\{f_n\}_{n=0}^\infty$ consider the Wronskian determinant of an arbitrary sequence $f_{k_1},\dots,f_{k_\ell}$ given by
\begin{equation}
f_\lambda:=\Wr[f_{k_1},\dots,f_{k_\ell}]=
\left| 
 \begin{array}{cccc} 
 f_1 & \cdots & f_\ell \\
  f_1'  & \cdots & f_\ell'\\
  \vdots  & & \vdots \\ 
  f_1^{(\ell-1)} &  \cdots & f_\ell^{(\ell-1)}
   \end{array} \right|
\end{equation}
where the multi-index $(k_1,\dots,k_\ell)$ is related to the partition $\lambda=(\lambda_1,\dots,\lambda_\ell)$ by 
\begin{equation}\label{eq:lambdadef}
 k_j=\lambda_j+j-1,\quad  j=1,...,\ell.
\end{equation}
Throughout the paper we will make use of the following notation:
\begin{eqnarray}\label{eq:flm}
f_{\lambda,m}&=&\Wr[f_{k_1},\dots,f_{k_\ell},f_m]\\
f_{\lambda,m,n}&=&\Wr[f_{k_1},\dots,f_{k_\ell},f_m,f_n] \label{eq:flmn}
\end{eqnarray}

\begin{definition}\label{def:nondeg}
A sequence of functions $\{f_n\}_{n=0}^\infty$ defined in $(a,b)$ is \textit{non-degenerate} if for every partition $\lambda$ and for every pair of integers $m,n$ the following two conditions are met
\begin{enumerate}
\item  $f_\lambda$ and $f_{\lambda,m}$ do not have a common root in $(a,b)$.
\item $f_{\lambda,m}$ and $f_{\lambda,n}$ do not have a common root in $(a,b)$.
\end{enumerate}
\end{definition}

We are now ready to state the main theorem of this paper:

\begin{thm}[Main theorem]\label{thm:main}
Let $\{\varphi_n\}_{n=0}^\infty$ be a non-degenerate sequence of eigenfunctions of a Hamiltonian $H$ as in $\eqref{eq:H}$.
Then the Wronskian determinant of an arbitrary sequence of eigenfunctions $\varphi_\lambda=\Wr[\varphi_{k_1},\dots,\varphi_{k_\ell}]$ has $n(\varphi_\lambda)$ simple real zeros in $(a,b)$, where
\begin{equation}\label{eq:altersum}
n(\varphi_\lambda)=\sum_{j=1}^\ell (-1)^{\ell-j} \lambda_j
\end{equation}
\end{thm}

\begin{remark}
We observe that this result includes Adler's Theorem \ref{thm:Adler} as a particular case since the sequences described there correspond via \eqref{eq:lambdadef} to  partitions of the form $(0,\dots,0,\lambda_1,\lambda_1,\dots,\lambda_\ell,\lambda_\ell)$ for which the alternate sum \eqref{eq:altersum} vanishes.
\end{remark}

A very frequent situation where the previous theorem does not hold occurs if the potential in \eqref{eq:H} is even, since in that case all odd eigenfunctions will vanish at zero and the sequence of eigenfunctions is degenerate. Definition \ref{def:nondeg} needs to be relaxed to include this case in which the high multiplicity root at $x=0$ will need a separate treatment.

\begin{definition}\label{def:semideg}
A sequence of functions $\{f_n\}_{n=0}^\infty$ defined in $(-a,a)$ is \textit{semi-degenerate} if for every partition $\lambda$ and for every pair of integers $m,n$ the following two conditions are met
\begin{enumerate}
\item  If $f_\lambda(x^*)=f_{\lambda,m}(x^*)=0$ then $x^*=0$.
\item If $f_{\lambda,m}(x^*)=f_{\lambda,n}(x^*)=0$ then $x^*=0$.
\end{enumerate}
\end{definition}
In other words, except for maybe at the origin, the sequence of eigenfunctions is non-degenerate.
In this symmetric case, the Wronskian $\varphi_\lambda$ has well defined parity
\begin{equation}\label{eq:parity}
\varphi_\lambda(-x)=(-1)^{|\lambda|} \varphi_\lambda(x)
\end{equation}
where $|\lambda|=\sum_{j=1}^\ell \lambda_j$. Moreover, to every partition $\lambda=(\lambda_1,\dots,\lambda_\ell)$ we can associate an integer $d_\lambda$ given by 
\begin{equation}\label{eq:ddef}
d_\lambda=p-q
\end{equation}
where $p$ and $q$ are the number of odd and even elements respectively in the sequence $k_1,\dots,k_\ell$, related to $\lambda$ by \eqref{eq:lambdadef}.

\begin{thm}\label{thm:symm}
Let $\{\varphi_n\}_{n=0}^\infty$ be the eigenfunctions of Schr\"odinger's equation \eqref{eq:H} with a symmetric potential $V(-x)=V(x)$ defined in $(-a,a)$.
If the sequence of eigenfunctions is semi-degenerate then the Wronskian determinant of an arbitrary sequence of eigenfunctions $\varphi_\lambda$ has
\begin{enumerate}
\item[i)] a root at $x=0$ of multiplicity $\frac{d_\lambda(d_\lambda+1)}{2}$
\item[ii)] $n_+(\varphi_\lambda)$ simple positive real roots, where $n_+(\varphi_\lambda)$ is given by
\begin{equation}\label{eq:n+}
n_+(\varphi_\lambda)=\frac{1}2\left(\sum_{i=1}^\ell (-1)^{\ell-i} \lambda_i-\frac{|d_\lambda+(\ell-2\lfloor\frac{\ell}{2}\rfloor)|}{2}\right)
\end{equation}
\item[iii)] the same number of negative real roots due to the symmetry \eqref{eq:parity}.
\end{enumerate}
\end{thm}

Note that Theorems \ref{thm:main} and \ref{thm:symm} apply to general eigenfunctions of an arbitrary potential, not necessarily polynomials.
However,  classical orthogonal polynomials fit naturally into this picture since up to a change of variable and multiplication by a nonzero pre-factor they are essentially the eigenfunctions of a Schr\"odinger problem \eqref{eq:H} for some very specific potentials.

Thus, it will not be difficult to derive the following corollaries concerning the number of zeros of an arbitrary Wronskian of classical orthogonal polynomials.
\begin{cor}\label{cor:Lag}
For almost every value of $\alpha\in(-1,\infty)$ the Wronskian of $\ell$ Laguerre polynomials $L_\lambda=\Wr\left[L_{k_1}^{(\alpha)},\dots,L_{k_\ell}^{(\alpha)}\right]$has $n(L_\lambda)$ simple zeros in $(0,\infty)$ where 
\begin{equation}\label{eq:zerosLag}
n(L_\lambda)=\sum_{j=1}^\ell (-1)^{\ell-j} \lambda_j.
\end{equation}
\end{cor}

\begin{cor}\label{cor:Jac}
If $\alpha\neq \beta$, for almost every value of $\alpha,\beta$ the Wronskian determinant of $\ell$ Jacobi polynomials $P_\lambda=\Wr\left[P_{k_1}^{(\alpha,\beta)},\dots,P_{k_\ell}^{(\alpha,\beta)}\right]$has $n(P_\lambda)$ simple zeros in $(-1,1)$ where 
\begin{equation}
n(P_\lambda)=\sum_{j=1}^\ell (-1)^{\ell-j} \lambda_j.
\end{equation}
If $\alpha=\beta$, then for almost every value of $\alpha$ the Wronskian determinant of $\ell$ Gegenbauer polynomials $C_\lambda=\Wr\left[P_{k_1}^{(\alpha,\alpha)},\dots,P_{k_\ell}^{(\alpha,\alpha)}\right]$ has 
\begin{enumerate}
\item[i)] a root at $x=0$ of multiplicity $\frac{d_\lambda(d_\lambda+1)}{2}$
\item[ii)] $n_+(C_\lambda)$ simple real roots in $(0,1)$ where $n_+(C_\lambda)$ is given by
\begin{equation}
n_+(C_\lambda)=\frac{1}2\left(\sum_{i=1}^\ell (-1)^{\ell-i} \lambda_i-\frac{|d_\lambda+(\ell-2\lfloor\frac{\ell}{2}\rfloor)|}{2}\right)
\end{equation}
\item[iii)] the same number of roots in $(-1,0)$ due to the symmetry $C_\lambda(-x)=(-1)^{|\lambda|} C_\lambda(x)$
\end{enumerate}
\end{cor}

The previous corollaries are a direct application of Theorems \ref{thm:main} and \ref{thm:symm}, after performing a gauge transformation and showing that for generic values of the parameters, these orthogonal polynomials are non-degenerate in the sense of Definition \ref{def:nondeg}. 

In the case of Hermite polynomials, we believe that Theorem \ref{thm:symm} holds but we lack a proof in the general case that the sequence $\{H_n\}_{n=0}^\infty$ is non-degenerate.
Even the simpler question of whether two Hermite polynomials can have a common root other than zero seems to be unanswered in the literature.
The non-degeneracy of  $\{H_n\}_{n=0}^\infty$ is intimately related to Conjecture 1 formulated by Felder et al. in \cite{Felder2012a} concerning the simplicity of all the roots (real and complex) of a Wronskian of Hermite polynomials $H_\lambda$.
Assuming non-degeneracy of Hermite polynomials, we can prove Conjecture 2 in \cite{Felder2012a} which gives the number of real and purely imaginary roots of the Wronskian of a double partition of Hermite polynomials.

It is noteworthy that Karlin and Szeg\H{o}'s Theorem \ref{thm:Karlin} applies only to sequences of consecutive orthogonal polynomials for which we would have $W(n,\ell,x)=P_\lambda$ for $\lambda=(n,n,\dots,n)$. The results for the number of zeros given by Karlin and Szeg\H{o}'s Theorem can thus be seen as a particular case of the alternate sum \eqref{eq:altersum}. However, the premises for Theorem \ref{thm:Karlin} require just an orthogonal polynomial system with respect to an arbitrary measure, so they need not satisfy a Sturm-Liouville system. The validity of formula \eqref{eq:altersum} seems thus to be larger, and it may apply to sequences of arbitrary orthogonal polynomials, not just to the classical or exceptional ones. It was already conjectured by Dur\'an in \cite{Duran2014a} that Theorem \ref{thm:Adler} applies not just to eigenfunctions but to arbitrary orthogonal polynomials. Our numerical explorations confirm this fact, and in fact allow us to formulate the following more general conjecture.
\begin{conj}\label{conj:OParb}
Let $\{P_n\}_{n=0}^\infty$ be an orthogonal polynomial system with respect to an arbitrary positive measure $d\mu=W\,dx$ supported on an interval $I$ of the real line. If  $\{P_n\}_{n=0}^\infty$ is non-degenerate in the sense of Definition \ref{def:nondeg} then the Wronskian determinant of an arbitrary sequence $P_\lambda$ has
\[n(P_\lambda)=\sum_{j=1}^\ell (-1)^{\ell-j} \lambda_j\]
simple real zeros in the support of $\mu$.
\end{conj}

\section{Zeros of Wronskians of eigenfunctions of Schr\"odinger's equation}
\label{sec:eigenfunctions}

In this section we will prove Theorems \ref{thm:main} and \ref{thm:symm} for the number of zeros of the Wronskian  of an arbitrary sequence of eigenfunctions of Schr\"odinger's equation. We will study all the possible degenerate cases appearing in the Wronskian of two and three eigenfunctions, and then we will introduce the non-degeneracy condition and prove the main theorems by induction.
The techniques involved in the proof are a refinement of those employed in \cite{Adler}, together with some algebraic identities satisfied by Wronskian determinants. Namely, we will use the following algebraic identity satisfied by the Wronskian determinant:
\begin{equation}\label{eq:Widentitya}
\Wr[\varphi_{k_1},\dots,\varphi_{k_\ell},\varphi_{k_{\ell+1}}, \varphi_{k_{\ell+2}}]=\frac{\Wr\left[\Wr[\varphi_{k_1},\dots,\varphi_{k_\ell},\varphi_{k_{\ell+1}}],\Wr[\varphi_{k_1},\dots,\varphi_{k_\ell},\varphi_{k_{\ell+2}}]\right]}{\Wr[\varphi_{k_1},\dots,\varphi_{k_\ell}]}.
\end{equation}
which can be written in compact terms as 
\begin{equation}\label{eq:Widentityb}
\Wljk=\frac{\Wr[\Wlj,\Wlk]}{\Wl}.
\end{equation}
using the notation introduced in \eqref{eq:flm}-\eqref{eq:flmn}.
First, we need to introduce some preliminary notions.

\begin{definition}
Given a function $f$ in $(a,b)$ and a point $x_0\in(a,b)$, let $\ord_{x_0}(f)$ denote the order of $x_0$ as a pole or root of $f$, i.e. $\ord_{x_0}(f)=1$ if $f$ has a simple root at $x=x_0$ and $\ord_{x_0}(f)=0$ if $x_0$ is neither a root nor a pole of $f$.
We also denote by $n(f)$ the number of times that $f$ vanishes in $(a,b)$ (not counting the multiplicities of the roots).
\end{definition}

\begin{definition}\label{def:ftypes}
Let  $f$ be a $C^2$ function defined in a punctured neighbourhood of $x_0\in \Rset$. We will consider the following possible behaviours of $f$ and its derivatives at $x_0$:
\begin{enumerate}
\item[] Type I:  $\,\,\,\,f(x_0)\neq0$,\quad $f'(x_0)= 0$,\quad $f''(x_0)=-kf(x_0)$ 
\item[] Type II: $\,\,f(x_0)\neq 0$,\quad  $f'(x_0)= 0$,\quad  $f''(x_0)= kf(x_0)$ 
\item[] Type III:   $f(x_0)=0$,\quad  $f'(x_0)= 0$,\quad  $f''(x_0)= 0$ 
\item[] Type IV:   $f$ has a pole at $x_0$.
\end{enumerate}
where $k$ is a positive constant.
\end{definition}

The following lemma will be necessary to count the number of zeros:

\begin{lem}
\label{lem:prevwr2}

Consider a function $f\in C^2([a,b])$ such that $f'$ does not vanish in $(a,b)$ and the behaviour of $f$ at $a$ and $b$ is of type I, II or III.
\begin{enumerate}
\item[i)] If $f$ is of type I at $a$ and $b$, the it has a simple zero in $(a,b)$.
\item[ii)] If $f$ is of type II or III at an endpoint, then it is of type I at the other endpoint and it has constant sign in $(a,b)$.
\end{enumerate}
\end{lem}

\begin{proof}
Without loss of generality we can assume $f$ to be increasing, i.e.
\begin{equation} \label{eq:inequalityw}
f(a)<f(b)
\end{equation} 
Since $f'(a)=f'(b)=0$ and  $f'(x)>0$ for $x\in(a,b)$, it follows that $f''(a)\geq0$ and $f''(b) \leq 0$. Therefore
\begin{equation} \label{eq:inequalitywpp}
f''(a)\geq f''(b).
\end{equation}
If $f$ is of type I at both endpoints, then $f$ and $f''$ have opposite signs at these points and we have $f(a)<0$ and $f(b)>0$, so $f$ must have a simple zero in $(a,b)$.
Note that $f$ must be of type I at least at one of the endpoints, since otherwise both inequalities  \eqref{eq:inequalityw}, \eqref{eq:inequalitywpp} cannot be simultaneously satisfied.
Without loss of generality, if $f$ is of type II or III at $a$, then $f(a)\geq 0$. At $b$, $f$ must be type I, so $f(b)>0$ and $f$ has constant sign in $(a,b)$.
\end{proof}

Adler proved \cite{Adler} that the Wronskian of two eigenfunctions has constant sign if and only if they are consecutive. A very straightforward generalization allows to prove the following proposition.

\begin{prop} 
\label{prop:wr2}
Let $\varphi_i$ and $\varphi_j$, $i<j$, be two eigenfunctions of the Schr\"odinger problem \eqref{eq:H}.
Then, \( \Wr[\varphi_i, \varphi_j]\) vanishes $j-i-1$ times inside $(a,b)$.  
If $x^*$ is one such root,  its multiplicity is 
\begin{equation}
\ord_{x^*}(\Wr[\varphi_i,\varphi_j])=\begin{cases}
3 &   \text{ if } \varphi_i(x^*)= \varphi_j(x^*)=0\\
1 &  \text{ otherwise.} 
\end{cases}
\end{equation}
Moreover, in the latter case neither $\varphi_i$ nor $\varphi_j$ vanish at $x^*$.
\end{prop}

\begin{proof}
Let us denote by
\begin{eqnarray}\label{eq:w}
w&=&\Wr[\varphi_i, \varphi_j]=\varphi_i \varphi'_j - \varphi'_i \varphi_j,\\
 \label{eq:wp}
w'&=& \varphi_i \varphi''_j - \varphi''_i \varphi_j =\delta \varphi_i\varphi_j,\\
\label{eq:wpp}
w''&=& -\delta (\varphi_i \varphi'_j + \varphi'_i \varphi_j) ,
\end{eqnarray}
where \(\delta=E_j -E_i >0. \)

Let $y_1<...<y_i$ and $x_1<...<x_j$  be the zeros of $\varphi_i$ and $\varphi_j$ in $(a,b)$. From standard Sturm Liouville theory we know that these zeros are simple and   \mbox{$x_1<y_k<x_j$ $\forall k=1,...,i$}. From \eqref{eq:wp} we see that these are the only zeros of $w'$ in $(a,b)$, and from \eqref{eq:w}-\eqref{eq:wpp} it is also clear that
\begin{eqnarray}
\label{eq:wppx} w''(y_k)&=&\delta w(y_k),\qquad k=1,\dots,i \\
\label{eq:wppy} w''(x_k)&=&-\delta w(x_k), \qquad k=1,\dots,j
\end{eqnarray}

We denote by $n_{ij}$ the number of common zeros of $\varphi_i$ and $\varphi_j$ in $(a,b)$. At each of these points, $z_k,k=1,\dots, n_{ij}$ it is clear that $w(z_k)=w'(z_k)=w''(z_k)=0$, but $w'''(z_k)\neq0$. Therefore $z_k$ is a triple root of $w$.

From Defintion \ref{def:ftypes} we see that:
\begin{enumerate}
\item $w$ is of type I at the points $x_k$ which are roots of $\varphi_j$ only.
\item $w$ is of type II at the points $y_k$ which are roots of $\varphi_i$ only.
\item $w$ is of type III at the common roots $z_k$ of $\varphi_i$ and $\varphi_j$,

\end{enumerate}

Note first that the roots of $w$ lie in $[x_1,x_j]$. Since $w'$ has a constant sign outside that interval and $w(a)=w(b)=0$, then $w$ cannot vanish in $(a, x_1)\bigcup(x_j, b)$. 
Let us count the number of roots that $w$ can have in $[x_1,x_j]$.  From Lemma \ref{lem:prevwr2} and the type I-III behaviour of $w$ at the roots of $\varphi_i$ and $\varphi_j$, we see that between every two consecutive roots of $\varphi_j$ there can be at most one root of $\varphi_i$. Suppose initially that $n_{ij}=0$, i.e. that all roots are distinct. Then there are  $j-1$ intervals $(x_k,x_{k+1}), k=1,\dots,j-1$, $i$ of which contain a root of $\varphi_i$ and $j-1-i$ that do not. Lemma  \ref{lem:prevwr2} asserts that in each of the  latter intervals there is exactly one simple root of $w$.
If some of the roots of $\varphi_i$ coincide with the roots of $\varphi_j$, Lemma \ref{lem:prevwr2} guarantees that the total number of points where $w$ vanishes does not change, but $w$ has a triple instead of a simple root at those points.
\end{proof}

\begin{remark}
Note that in the coalescence process where one root of $\varphi_i$ approaches another root of $\varphi_j$, a simple root (which is different from the previous ones) and two complex conjugate roots of $\Wr[\varphi_i,\varphi_j]$ meet at the coalescence point giving a third order root.
\end{remark}

The following Lemma extends  Lemma \ref{lem:prevwr2} to the case when the function is allowed to have a pole at the endpoints of the interval.

\begin{lem}
\label{lem:prevwr3}
Consider a function $f\in C^2((a,b))$ such that $f'$ does not vanish in $(a,b)$.
We assume that $f$ is of type IV at one of the endpoints.
\begin{enumerate}
\item[i)] If $f$ is of type I or IV at the other endpoint, then $f$ has exactly one simple zero in $(a,b)$.
\item[ii)] If $f$ is of type II or III at the other endpoint, then $f$ has constant sign in $(a,b)$.
\end{enumerate}
\end{lem}
\begin{proof}
$f$ is monotonic in $(a,b)$. If it is of type IV at both endpoints, then it must approach $+\infty$ and $-\infty$ at the endpoints, and therefore it has exactly one simple zero in $(a,b)$.

Without loss of generality we assume that $f$ is of type IV at $b$. If $f$ is of type I at $a$, then $f(a)$ and $f'(x), \forall x\in (a,b)$ have opposite sign. By continuity, $f$ has exactly one simple zero in $(a,b)$. If $f$ is of type II at $a$, then $f(a)$ and $f'(x), \forall x\in (a,b)$ have the same sign, so $f$ does not vanish in $(a,b)$. The same is true if $f$ is of type III at $a$: since it is strictly monotonic and it vanishes at $a$, it has constant sign in $(a,b)$.
\end{proof}

In order to extend the previous result to the Wronskian of three eigenfunctions we recall that if $\varphi_i$, $\varphi_j$ are eigenfunctions of Schr\"odinger's equation \eqref{eq:H} with a regular potential $V$, then the following functions
\begin{equation}\label{eq:bphi}
\bphi_{j}=\frac{\Wr[\varphi_i, \varphi_j]}{\varphi_i},\qquad j\neq i
\end{equation}
satisfy Schr\"odinger's equation:
\begin{equation}
-\bphi_j''+ \bar V \bphi_j=E_j \bphi_j
\end{equation}
for the transformed potential
\begin{equation}
\bar{V} =V-2D_{xx}(\log \varphi_i).
\end{equation}
Note that the potential $\bar V$ will have double poles at the zeros of $\varphi_i$ and we see thus the necessity of Lemma \ref{lem:prevwr3} to treat the behaviour of the eigenfunctions at the poles of the potential.

We are now ready to state the result for the Wronskian of three eigenfunctions

\begin{prop}\label{prop:3phis}
Let $\varphi_i,\varphi_j$ and $\varphi_k$  with $i<j<k$ be three eigenfunctions of \eqref{eq:H}. The Wronskian determinant $\Wr[\varphi_i,\varphi_j,\varphi_k]$  vanishes exactly $k-j+i-1-n_{ijk}$ times in $(a,b)$, where $n_{ijk}$ is the number of simultaneous  roots of $\varphi_i$, $\varphi_j$ and $\varphi_k$ in $(a,b)$. Let $x^*$ be one such root, then its multiplicity is given by 
\begin{equation} \label{eq:ordthm1}
\ord_{x^*}(\Wr[\varphi_i,\varphi_j,\varphi_k])=\begin{cases}
6 & \text{ if } \varphi_i(x^*)=\varphi_j(x^*)=\varphi_k(x^*)=0 \\
3 &   \text{ if } \varphi_i(x^*)\neq 0 \text{ and } \bphi_j(x^*)=\bphi_k(x^*)=0 \\
1 &  \text{ otherwise, }
\end{cases}
\end{equation}
where $\bphi_j=\varphi_i^{-1}\,\Wr[\varphi_i,\varphi_j]$ and $\bphi_k=\varphi_i^{-1}\,\Wr[\varphi_i,\varphi_k]$.
\end{prop}

\begin{proof}
The identity for Wronskian determinants \eqref{eq:Widentitya} in the case of three functions reads
\begin{equation}\label{eq:W3phis}
\Wr[\varphi_i, \varphi_j, \varphi_k]=\frac{\Wr[\Wr[\varphi_i, \varphi_j], \Wr[\varphi_i, \varphi_k]]}{\varphi_i}.
\end{equation}
which can be rewritten as
\begin{equation}\label{eq:wr3w}
\Wr[\varphi_i, \varphi_j, \varphi_k]
  =\varphi_i w,
\end{equation}
where
\begin{equation}
w=\Wr[\bphi_{j}, \bphi_{k}].
\end{equation}
The derivatives of $w$ obey the relations
\begin{equation}\label{eq:wbar}
w'= \bphi_j\bphi_k''-\bphi_j''\bphi_k=-\delta \bphi_j\bphi_k,\qquad
w''=-\delta (\bphi_j\bphi_k'+\bphi_j'\bphi_k), \qquad \delta=E_k-E_j>0.
\end{equation}
We shall first count the number of times that $w$ vanishes in $(a,b)$. To this end, we need to consider the behaviour of $w$ and its derivatives at each of the points where $\varphi_i$, $\bphi_j$ and $\bphi_k$ vanish. We shall denote by $n_{ij}$ the number of common roots of $\varphi_i$ and $\varphi_j$ where $\varphi_k$ does not vanish, and $n_{ijk}$ the number of points where all three functions vanish. Likewise, $\bar n_{jk}$ denotes the number of common roots of $\bphi_j$ and $\bphi_k$ where $\varphi_i$ is not zero. Using expressions \eqref{eq:wr3w}, \eqref{eq:wbar} and Proposition \ref{prop:wr2}, it is not hard to derive the results gathered in Table 1.

\begin{table}[ht]\label{table1}
\begin{tabular}{|c|c|c|c|c|c|c|c|}
\hline 
$\varphi_i(x^*)$ & $\bphi_j(x^*)$& $\bphi_k(x^*)$ & $\ord_{x^*} \bphi_j$ & $\ord_{x^*} \bphi_k$ & $\ord_{x^*} w$ & type$_{x^*}$ $w$    & \# points  \\ 
\hline 
0 & $\neq 0$ & $\neq 0$ &  $\!\!\!\!\!-1$ & $\!\!\!\!\!-1$ &$\!\!\!\!\!-1$&   IV & $i-n_{ij}-n_{ik}-n_{ijk}$  \\ 
\hline 
$\neq 0$ & $0$ & $\neq 0$ &  $1$ & $0$ &0&   II & $j-i-1-n_{ij}-n_{ijk}-\bar n_{jk}$  \\ 
\hline 
$\neq 0$ & $\neq 0$ & $0$ &  $0$ & $1$ &0&   I & $k-i-1-n_{ik}-n_{ijk}-\bar n_{jk}$  \\ 
\hline 
$0$ & $0$ & $\neq 0$ &  $2$ & $\!\!\!\!\!-1$ &0&   II & $n_{ij}$  \\ 
\hline 
$0$ & $\neq 0$ & $0$ &  $\!\!\!\!\!-1$ & $2$ &0&   I & $n_{ik}$  \\ 
\hline 
$\neq 0$ & $0$ & $0$ &  $1$ & $1$ &3&   III & $\bar n_{jk}$  \\ 
\hline 
$0$ & $0$ & $0$ &  $2$ & $2$ &5&   III & $n_{ijk}$  \\ 
\hline 
\end{tabular}

\caption{Behaviour of $w$ at the roots of $\varphi_i$, $\bphi_j$ and $\bphi_k$} 
\end{table}

For illustrative purposes, we explain the fourth row in Table 1. If $\bphi_j(x^*)=0$ and $\varphi_i(x^*)=0$, then by \eqref{eq:bphi} $\varphi_j(x^*)=0$. By Propositon \ref{prop:wr2} $\Wr[\varphi_i, \varphi_j]$ has a triple root at $x^*$, so $\ord_{x^*} \bphi_j=2$. It is also clear that $\bphi_k$ must have a simple pole at $x^*$. 
We can thus write 
\[ \bphi_j=A(x)(x-x^*)^2,\qquad \bphi_k=\frac{B(x)}{x-x^*}\]
where $A(x)$ and $B(x)$ are regular and do not vanish at $x^*$. From \eqref{eq:wr3w} and \eqref{eq:wbar} we have
\begin{equation}
w(x^*)=-3A(x^*)B(x^*)  \qquad w''(x^*)=-\delta A(x^*)B(x^*)
\end{equation}
so $w$ is of type II at $x^*$. The rest of the entries in the table can be derived in a similar manner.

Suppose initially that $n_{ij}=n_{ik}=n_{ijk}=\bar n_{jk}=0$, i.e. that all roots are distinct. Let us first consider the roots of $\varphi_i$ and $\bphi_k$. Since by assumption $\bphi_j$ does not vanish in those points, from Table 1 we see that $w$ is of type I or IV there. There are exactly $n(\varphi_i)=i$ roots of the first kind and $n(\bphi_k)=k-i-1$ roots of the second type. Let us denote by $x_{\min}$ the smallest of such roots and $x_{\max}$ the largest. Since $\bphi_j$ and $\bphi_k$ vanish at the endpoints $a$ and $b$, then $w$ is of type III at those points. By Lemmas \ref{lem:prevwr2} and \ref{lem:prevwr3} there can be no roots of $\bphi_j$ and $w$ in $(a,x_{\min})\cup(x_{\max},b)$. The interval $[x_{\min},x_{\max}]$ contains all the roots of $w'$, so we can divide it into $n(\bphi_k)+n(\varphi_i)-1$ intervals, limited by the previous roots, where $w$ is monotonic. 

We now consider the position of the $n(\bphi_j)=j-i-1$ roots of $\bphi_j$ in $[x_{\min},x_{\max}]$. By Lemma \ref{lem:prevwr2}, between any two such roots there must be at least one root of $\varphi_i$ or $\bphi_k$. Out of the previous $n(\bphi_k)-n(\varphi_i)-1$ intervals, there are $n(\bphi_j)$ that contain exactly one root of $\bphi_j$. Lemmas \ref{lem:prevwr2} and \ref{lem:prevwr3} ensure that in these intervals there are no roots of $w$. The remaining $n(\bphi_k)-n(\bphi_j)-n(\varphi_i)-1$ intervals are limited by points where $w$ is of type I or IV and by \eqref{eq:wbar} $w'$ does not vanish inside them. By Lemmas \ref{lem:prevwr2} and \ref{lem:prevwr3} we conclude that $w$ has exactly  $n(\bphi_k)-n(\bphi_j)+n(\varphi_i)-1=k-j+i-1$ simple zeros, one in each of these intervals.
 
If some of the roots of $\bphi_j$ coincide with roots of $\bphi_k$ but not with roots of $\varphi_i$, Lemma \ref{lem:prevwr2} guarantees that the total number of points where $w$ vanishes does not change. 
If a root of $\varphi_i$ coincides with a root of $\bphi_j$  or $\bphi_k$ (or both),  then $w$ has one root less in $(a,b)$ with respect to the case where all roots are distinct.

We conclude then that the total number of roots that $w$ has in $(a,b)$ in the case where the roots of  $\varphi_i$ , $\bphi_j$  and $\bphi_k$ are allowed to coincide is:
\begin{equation} \label{eq:n(w)}
n(w)=n(\bphi_k)-n(\bphi_j)+n(\varphi_i)-1-n_{ij}-n_{ik}-n_{ijk},
\end{equation}
which evaluates to
\begin{equation}
n(w)=k-j+i-1-n_{ij}-n_{ik}-n_{ijk}.
\end{equation}

From \eqref{eq:wr3w} we see that these are not the only roots of $\Wr[\varphi_i,\varphi_j,\varphi_k]$: we need to consider also the roots $x^*$ of $\varphi_i$ that are not roots of $w$, i.e. those for which  $\ord_{x^*} w =0$. A look at Table 1 suffices to conclude that the total number of roots of $\Wr[\varphi_i,\varphi_j,\varphi_k]$ in $(a,b)$ is $n(w)+n_{ij}+n_{ik}$ which becomes $k-j+i-1-n_{ijk}$.
The multiplicity of the roots of $\Wr[\varphi_i,\varphi_j,\varphi_k]$ can be easily calculated using Lemmas \ref{lem:prevwr2} and \ref{lem:prevwr3} and expressions \eqref{eq:wr3w}-\eqref{eq:wbar}. This concludes the proof of Proposition \ref{prop:3phis}. 

\end{proof}

In order to prove Theorems \ref{thm:main} and \ref{thm:symm}, we first observe that the functions
\begin{equation}\label{eq:bphil}
\bphi_{j}:=\frac{\Wlj}{\Wl}= \frac{\Wr[\varphi_{k_1},\dots,\varphi_{k_\ell},\varphi_j]}{\Wr[\varphi_{k_1},\dots,\varphi_{k_\ell}]},\qquad j\neq{k_1,\dots, k_\ell}
\end{equation}
satisfy the equation:
\begin{equation}
-\bphi_j''+ \bar V_\lambda \bphi_j = E_j\bphi_j
\end{equation}
for the transformed potential
\begin{equation}
\bar{V}_\lambda =V-2D_{xx}(\log \Wl).
\end{equation}
Note that $\bphi_j$ are not true eigenfunctions of a well defined Schr\"odinger's problem because the potential $\bar V _\lambda$ has poles at the zeros of $\Wl$. 
We can rewrite \eqref{eq:Widentityb} as 
\begin{equation}\label{eq:Wljkw}
\Wljk  =\Wl w,
\end{equation}
where
\begin{equation}\label{eq:wdef}
w=\Wr[\bphi_{j}, \bphi_{k}].
\end{equation}
and $\bphi_j$ and $\bphi_k$ are defined by \eqref{eq:bphil}. It is straightforward to check that the derivatives of $w$ obey the relations \eqref{eq:wbar}.

\begin{proof}[Proof of Theorem \ref{thm:main}]
Let us prove the statement  by induction.
For $\ell=1$, $\varphi_{k_1}$  has $n(\varphi_{k_1})=k_1=\lambda_1$ simple real roots in $(a,b)$.  For $\ell=2$, from Proposition \ref{prop:wr2}, 
\[n(\Wr[\varphi_{k_1}, \varphi_{k_2}])=k_2-k_1-1=\lambda_2-\lambda_1.\] In addition, since $\varphi_{k_1}$ and $\varphi_{k_2}$ have no common roots, all the roots of $\Wr[\varphi_{k_1}, \varphi_{k_2}]$ are simple.
Formula \eqref{eq:altersum} holds thus for $\ell=1,2$. Next, we assume  that  formula \eqref{eq:altersum} holds for $\Wl$ and $\Wlj$ and all the roots in $(a,b)$ are simple.
We will use the same  argument as in the proof of Theorem \ref{prop:3phis} where $\Wl$  plays the role of $\varphi_i$ and $\bphi_j$ and $\bphi_k$ have the same role. Since $\{\varphi_n\}_{n=0}^\infty$ are non-degenerate by hypothesis, $\bphi_j$ has zeros at all those points where $\Wlj$ vanishes, i.e. $n(\bphi_j)=n(\Wlj)$ and $n_{ij}=n_{ik}=n_{ijk}=0$.  Using equation \eqref{eq:n(w)} we can write
\begin{equation}
n(w)=n(\Wlk)-n(\Wlj)+n(\Wl)-1.
\end{equation}
Since $\lambda_{\ell+1}=j-\ell$ and $\lambda_{\ell+2}=k-\ell-1$ and we assume that  \eqref{eq:altersum} holds for $\Wl$ and $\Wlj$, the previous expression evaluates to
\begin{equation}
n(w)=\sum_{i=1}^{\ell+2} (-1)^{\ell+2-1}\lambda_i
\end{equation}
From \eqref{eq:Wljkw} we finally conclude that $n(\Wljk)=n(w)$ since the simple zeros of $\varphi_\lambda$ are simple poles of $w$ due to the non-degeneracy condition.
Non-degeneracy also implies that all the roots of $\Wljk$ in $(a,b)$ are simple, which closes the induction and proves the desired result.
\end{proof}

\begin{proof}[Proof of Theorem \ref{thm:symm}]
Let us first prove by induction that 
\begin{equation}\label{eq:ord0}
\ord_{0}(\Wl)=\frac{\dl(\dl+1)}{2}.
\end{equation} 
Let $\ell$ be the length of $\lambda$. For $\ell=1$ it is clear that  $\varphi_{k_1}$ has a simple root at $x=0$ if $\dl=1$ and no root at $x=0$ otherwise ($\dl=-1$), so \eqref{eq:ord0} holds for $\ell=1$. For $\ell=2$, if $\dl=2$ or $\dl=0$ formula \eqref{eq:ord0} holds directly from Proposition \ref{prop:wr2}. If $\dl=-2$, i.e. $\varphi_{k_1}$ and $\varphi_{k_2}$ are odd, then their derivatives have a common root at $x=0$ and $\Wr[\varphi_{k_1}, \varphi_{k_2}]$  has a simple root at $x=0$. Therefore, formula \eqref{eq:ord0} holds also for $\ell=2$. 
Next, we will prove that if expression \eqref{eq:ord0} holds for $\Wl$ and $\Wlj$, it also holds for $\Wljk$ for any $\lambda$, $j$ and $k$.
To this end, first note that
\begin{equation}
d_{\lambda,j}=\begin{cases}
d_\lambda+1 &\text{if $j$ is odd,}\\
d_\lambda-1 &\text{if $j$ is even.}
\end{cases}
\end{equation}
Therefore, since \eqref{eq:ord0} holds for $\Wlj$ we have
\begin{equation}
\ord_0(\Wlj)=\begin{cases}
\ord_0(\Wl)+d_\lambda+1 &\text{if $j$ is odd,}\\
\ord_0(\Wl)-d_\lambda &\text{if $j$ is even.}
\end{cases}
\end{equation}
We will next show that \eqref{eq:ord0} holds for $\Wljk$ closing the induction.
All the possible cases are summarized in the following table: 
\begin{table}[h]
\begin{tabular}{|c|c|c|c|c|c|c|}
\hline
$j$ & $k$ &$d_{\lambda,j,k}$ & $\ord_0(\bphi_j)$ & $\ord_0(\bphi_k)$ &$\ord_0(w)$ & $\ord_0(\Wljk)$\\ \hline
odd & odd & $d_\lambda+2$ & $d_\lambda+1$ & $d_\lambda+1$ & $2d_\lambda+3$ & $\ord_0(\Wl) +2d_\lambda+3=\frac{(d_{\lambda,j,k}+2)(d_{\lambda,j,k}+3)}{2}$\\ \hline
odd & even &  $d_\lambda$ &$d_\lambda+1$ & $-d_\lambda$ & $0$ & $\ord_0(\Wl) =\frac{d_{\lambda,j,k}(d_{\lambda,j,k}+1)}{2}$\\ \hline
even & even & $d_\lambda-2$ & $-d_\lambda$ & $-d_\lambda$ & $-2\dl+1$ & $\ord_0(\Wl) -2d_\lambda+1= \frac{(d_{\lambda,j,k}-2)(d_{\lambda,j,k}-1)}{2}$\\ \hline
\end{tabular}

\caption{}
\end{table}

Note that from \eqref{eq:bphil} and \eqref{eq:Wljkw},  $\ord_{0}(\bphi_j)=\ord_{0}(\Wlj)- \ord_{0}(\Wl)$ and $\ord_{0}(\Wljk)=\ord_{0}(w)+ \ord_{0}(\Wl)$, while $\ord_{0}(w)$ can be obtained  from  \eqref{eq:wbar} and \eqref{eq:wdef} and taking into account the fact that $\Wlj$ and $\Wlk$ have a well defined parity.

Let us prove by induction that the number of times that $\Wl$ vanishes in the real interval $(-a,a)$ is given by 
\begin{equation}\label{eq:nW}
n(\Wl)=2n_+(\Wl)+\begin{cases}
0 & \text{if } \dl=-1, \dl=0\\
1 & \text{if } \dl \neq -1, 0
\end{cases}
\end{equation}
where we have used formula \eqref{eq:ord0} and $n_+(\Wl)$ is given by \eqref{eq:n+}. In addition, we will show that all those roots except maybe $x=0$ are simple.
For $\ell=1$, $\varphi_{k1}$ has $k_1=\lambda_1$ simple real roots. Since $\dl=1$ or $\dl=-1$, formula \eqref{eq:nW} holds for $\ell=1$.
For $\ell=2$, it can be proved directly from Proposition \ref{prop:wr2}.
Next, we assume that formula \eqref{eq:nW}  holds for $\Wl$ and $\Wlj$ and all real roots are simple except maybe $x=0$.
Again, we will use the same  argument as in the proof of Theorem \ref{prop:3phis} where $\Wl$  plays the role of $\varphi_i$ and $\bphi_j$ and $\bphi_k$ have the same role.
Since at $x=0$ we have the same behaviour as the one of Table 1, we can apply equation \eqref{eq:n(w)} where 
\begin{equation}
n_{ij}=\begin{cases}
1 & \text{if } \ord_{0}(\Wl)>0, \quad \ord_{0}(\bphi_j)>0, \quad \ord_{0}(\bphi_k)\leq 0 \\
0 & \text{otherwise}
\end{cases}
\end{equation} 
\begin{equation}
n_{ik}=\begin{cases}
1 & \text{if } \ord_{0}(\Wl)>0, \quad\ord_{0}(\bphi_k)>0, \quad \ord_{0}(\bphi_j)\leq 0 \\
0 & \text{otherwise}
\end{cases}
\end{equation} 
\begin{equation}
n_{ijk}=\begin{cases}
1 & \text{if } \ord_{0}(\Wl)>0,\quad \ord_{0}(\bphi_j)>0,\quad \ord_{0}(\bphi_k)>0 \\
0 & \text{otherwise}
\end{cases}
\end{equation}
Moreover, all the roots different from $x=0$ are simple. In order to take into account the total number of real roots of $\Wljk$ we must add to $n(w)$  one more root at $x=0$ if  $\ord_{0}(w)\leq 0$ and $\ord_0(\Wljk)>0$, which happens when 
\begin{itemize}
\item[i)] $\ord_0(w)=0$ and $\dl\neq -1,0$: $j$ and $k$ have different parity.
\item[ii)] $\ord_0(\Wl)>-\ord_0(w)>0$: $j$ and $k$ are odd and $\dl<-3$ or $j$ and $k$ are even and $\dl>2$.
\end{itemize} 
Thererefore
\begin{equation}
n(\Wljk)=n(\Wl)+n(\bphi_k)-n(\bphi_j)-1+m_{\lambda,j,k}
\end{equation}
where
\begin{equation}
m_{\lambda,j,k}=-n_{ijk}+
\begin{cases}
1 & \text{if $j$ and $k$ are odd and } \dl<-3 \\
1 & \text{if $j$ and $k$ are even and } \dl>2 \\
0 & \text{otherwise}
\end{cases}
\end{equation}
Applying this formula for all the possible parities of $j$ and $k$ and values of $\dl$ we obtain that formula \eqref{eq:nW} holds for $\ell+2$.

For instance, let us see in more detail the case when $\ord_{x_0}(\varphi_j)=\ord_{x_0}(\varphi_k)=1$. We have to take into account that 
\begin{equation}
2n_+(\Wlk)-2n_+(\Wlj) +2n_+(\Wl)-1=
\sum_{i=1}^{\ell+2} (-1)^{\ell+2-i} \lambda_i-\bigg| \frac{\dl+(\ell-2\lfloor\frac{\ell}{2}\rfloor)}{2}\bigg|
\end{equation}
where $\lambda_{\ell+1}=j-\ell$ and $\lambda_{\ell+2}=k-\ell-1$.
\begin{table}[ht]
\begin{tabular}{|c|c|c|c|c|c|}
\hline
$d_\lambda$& $n(\Wl)$& $n(\bphi_j)$& $n(\bphi_k)$&$m_{\lambda,j,k}$& $n(\Wljk)$ \\ \hline
$0$ & $2n_+(\Wl)$& $2n_+(\Wlj)+1$& $2n_+(\Wlk)+1$& $0$ &$2n_+(\Wljk)+1$  \\ \hline
$>0$ & $2n_+(\Wl)+1$& $2n_+(\Wlj)+1$& $2n_+(\Wlk)+1$& $\!\!\!\!-1$ & $2n_+(\Wljk)+1$  \\ \hline
$-1$ & $2n_+(\Wl)$& $2n_+(\Wlj)$& $2n_+(\Wlk)$& $0$ & $2n_+(\Wljk)+1$  \\ \hline
$-2, -3$ & \multirow{2}{*}{$2n_+(\Wl)+1$}& \multirow{2}{*}{$2n_+(\Wlj)$}& \multirow{2}{*}{$2n_+(\Wlk)$}& 0 &$2n_+(\Wljk)$\\ 
 $<-3$& & & & 1 &$2n_+(\Wljk)+1$  \\ \hline
\end{tabular}
\caption{}
\end{table}

\end{proof}

\section{Zeros of Wronskians of classical orthogonal polynomials}
\label{sec:OPs}

In this section we apply the Theorems \ref{thm:main} and \ref{thm:symm} proved in the previous Section to derive some results on the zeros of the Wronskian of classical orthogonal polynomials in their interval of orthogonality.
The first key observation is that classical orthogonal polynomials are essentially (up to a pre-factor and a change of variable) the eigenfunctions of a Schr\"odinger problem \eqref{eq:H}.
For every family of classical orthogonal polynomials $\{P_n\}_{n=0}^ \infty$ we can write
\begin{equation}\label{eq:phitoP}
\varphi_n(x)= \mu(x) P_n\big(z(x)\big) 
\end{equation}
where $\varphi_n(x)$ is the set of eigenfunctions of a Schr\"odinger problem \eqref{eq:H} for a given potential $V(x)$. More specifically, the three families of classical orthogonal polynomials are gathered in Table 4.
\begin{table}[ht]\label{table:clOP}
\begin{tabular}{|c|c|c|c|c|c|}
\hline
Class & $V(x)$ &  $(a,b)$ &  $\mu(x)$ & $P_n(z)$ & $z(x)$  \\
\hline
Hermite & $x^2$ & $(-\infty,\infty)$ & ${\rm e}^{-x^2/2}$ & $H_n(z)$ & $x$ \\
 \hline
Laguerre & $x^2+ \frac{\alpha^2-1/4}{x^2}$ & $(0,\infty)$ & $x^{\alpha+1/2}{\rm e}^{-x^2/2}$ & $L^{(\alpha)}_n(z)$ & $x^2$ \\
\hline
Jacobi & $\frac{\alpha^2-1/4}{\sin^2 (x-\pi/4)}+ \frac{\beta^2-1/4}{\cos^2 (x-\pi/4)}   $& $(-\frac{\pi}4,\frac{\pi}4)$  &$\big(\sin(x-\frac{\pi}4)\big)^{\alpha+1/2}\big(\cos (x-\frac{\pi}4)\big)^{\beta+1/2}$ & $P^{(\alpha,\beta)}_n(z)$ & $\sin 2x$ \\
\hline
\end{tabular}

\caption{}
\end{table}

It is also clear that $z'(x)>0$ and $\mu(x)>0$ for all $x\in(a,b)$, so it is not difficult to prove the following Proposition.

\begin{prop}\label{prop:WrOP}
For any given partition $\lambda=(\lambda_1,\dots,\lambda_\ell)$, if the functions $\varphi_n(x)$ are related to the polynomials $P_n(z)$ by \eqref{eq:phitoP}, then we have the following identity among their Wronskian determinants:
\begin{equation}\label{eq:WrphiP}
\Wr[\varphi_{k_1}(x),\dots,\varphi_{k_\ell(x)}]=\mu(x)^\ell\,\left(\frac{dz}{dx}\right)^\frac{\ell(\ell-1)}{2} \Wr[P_{k_1}(z),\dots,P_{k_\ell}(z)] 
\end{equation}

\end{prop}
We see in particular that the number and multiplicity of zeros of $\varphi_\lambda(x)$ in $(a,b)$ coincides with the those of $P_\lambda(z)$ in $(z(a),z(b))$.

\begin{proof}
Note that we can write the $n^{\rm th}$ derivative of $\varphi_k$ as
\begin{equation}\label{eq:phideriv}
\varphi_k^{(n)}=(z')^n \mu \, P_k^{(n)}+ \sum_{j=0}^{n-1} f_{n,j}(\mu,z') P^{(j)}_k
\end{equation}
where $f_{n,j}(\mu,z')$ is a polynomial expression in $\mu$, $z'$ and its derivatives up to order $n$.
Inserting this expression on the left hand side of \eqref{eq:WrphiP} and using the invariance of the determinant under linear combinations of its columns we see that only the first term in \eqref{eq:phideriv} matters, which leads to \eqref{eq:WrphiP}.
\end{proof}

The above proposition allows for a direct application of the formulas \eqref{eq:altersum} and $\eqref{eq:n+}$ in Theorems \ref{thm:main} and \ref{thm:symm} to the case of classical orthogonal polynomials. In order to complete the proof of Corollary \ref{cor:Lag} and \ref{cor:Jac} we need to discuss the degeneracy or semi-degeneracy of these families according to Definitions \ref{def:nondeg} and  \ref{def:semideg}.

Let us discuss first the non-degeneracy condition on Laguerre polynomials. The $k$ zeros of any given $L_k^{(\alpha)}$ move monotonously as $\alpha$ increases ,\cite{Dimitrov}. It is natural to expect that for some values of $\alpha$, two zeros of $L_k^{(\alpha)}$ and $L_j^{(\alpha)}$ coincide. As a matter of fact, the resultant of two such polynomials is a polynomial expression in $\alpha$, so it vanishes only for a finite number of values of $\alpha$. Repeating this argument for every possible pair of Laguerre polynomials, the values of $\alpha$ for which they have a common zero is a numerable set. The same argument can be extended to Wronskian determinants of any sequence of Laguerre polynomials to conclude that the set of $\alpha$ values for which the sequence $\{L_n^{(\alpha)}\}_{n=0}^\infty$ is non-degenerate is a numerable set.
This observation, together with \eqref{eq:WrphiP} implies that Corollary \ref{cor:Lag} follows from Theorem \ref{thm:main}.

For Jacobi polynomials two cases need to be distinguished. If $\alpha\neq\beta$ then the potential has no symmetry while for ultraspherical polynomials ($\alpha=\beta$) the potential is even. Similar arguments on the degeneracy as those presented above imply that Corollary \ref{cor:Jac} follows from Theorem \ref{thm:main} and  \ref{thm:symm}.

\subsection{Hermite polynomials}\label{sec:Hermite}
In the case of Hermite polynomials, we do not have a proof that $\{H_n\}_{n=0}^\infty$ is nongenerate in the sense of Definition \ref{def:nondeg}. This question is connected to a long standing conjecture in the theory of monodromy free potentials, (see Conjecture 1 in \cite{Felder2012a}):
\begin{conj}
For every partition $\lambda$, all the zeros (real and complex) of $H_\lambda$ are simple, except maybe $x=0$.
\end{conj}
Numerical explorations show a very strong support for this conjecture, but at the moment we lack a proof in the general case. We have at least a proof that holds for the case of partitions of length $2$:
\begin{prop}\label{prop:easy-Ves}
 The (real and complex) roots of $\Wr[H_m,H_n]$ are all simple, except $x=0$ that has multiplicity $3$ if both $m,n$ are odd.
\end{prop}

Before we can attempt the proof of this last Proposition, we need to establish two previous Lemmas.

\begin{lem}\label{lem:H}
If two Hermite polynomials satisfy $H_m(x^*)=H_n(x^*)=0$ then $m,n$ are odd and $x^*=0$.
\end{lem}

\begin{proof}
This proof uses a result by Schur on the irreducibility of Hermite polynomials, which states that an Hermite polynomial cannot be factored into two polynomials with rational coefficients, \cite{Schur,Dorwart1935}.
Hermite polynomials have rational coefficients, which means that the greatest common divisor (GCD) of any two Hermite polynomials also has rational coefficients, since it can be computed using Euclides' algorithm.
We argue by contradiction: suppose there exits $x^*\in\mathbb R-\{0\}$ such that $H_m(x^*)=H_n(x^*)=0$. Then we have the following factorization
\[ H_m(x)=A(x) P(x),\quad H_n(x)=A(x)Q(x)\]
where $A(x)=GCD(H_m,H_n)$ is a polynomial of degree at least one and $A(x)\neq x$.
From the argument above, $A(x)$ should have rational coefficients, which is in contradiction with the irreducible character of Hermite polynomials.

\end{proof}

\begin{lem} \label{lem:simple}
Two Hermite polynomials $H_m$ and $H_n$ do not have a root in common if and only if $\Wr[H_m, H_n]$ has simple roots. If $H_m(x^*)=H_n(x^*)=0$ then $x^*$ is a triple root of $\Wr[H_m, H_n]$.
\end{lem}

\begin{proof}
Let $w(x)=\Wr[H_m,H_n]$. We have the following expressions
\begin{eqnarray}
w(x)&=& H_mH_n'-H_m'H_n \label{eq:wH0}\\
w'(x)&=&H_m H_n''-H_n'' H_m \label{eq:wH1}
\end{eqnarray}

Suppose that $x^*$ is a common root of $H_m$ and $H_n$. It is obvious that $H_m'(x^*)\neq 0$ and $H_n'(x^*)\neq 0$. Taking derivatives of $w(x)=\Wr[H_m,H_n]$ it is not hard to verify that $w(x^*)=w'(x^*)=w''(x^*)=0$ but $w'''(x^*)\neq0$.
It remains to be proven that if $x^*$ is a root of $w(x)$ of multiplicity two or higher, then $H_m(x^*)=H_n(x^*)=0$. We assume that $w(x^*)=w'(x^*)=0$. 
 Suppose initially that neither $H_m$ nor $H_n$ vanish at $x^*$.
From \eqref{eq:wH0} and \eqref{eq:wH1} we have
\begin{equation} \label{eq:rel}
 \frac{H'_m}{H_m}\bigg |_{x^*}=\frac{H'_n}{H_n}\bigg |_{x^*} \hspace*{2 cm} \frac{H''_m}{H_m}\bigg |_{x^*}=\frac{H''_n}{H_n}\bigg |_{x^*}.
\end{equation} 
Using the differential equation satisfied by Hermite polynomials:
\begin{equation} \label{eq:DEHermite}
 H_n''-2xH_n'+2nH_n=0
 \end{equation}
we see that
\begin{equation}
\frac{H''_m}{H_m}\bigg |_{x^*}-2x^*\frac{H'_m}{H_m}\bigg |_{x^*}=-2m=
\frac{H''_n}{H_n}\bigg |_{x^*}-2x^*\frac{H'_n}{H_n}\bigg |_{x^*}=-2n.
\end{equation}
which leads to a contradiction since we assume that $m\neq n$.
The only possibility is that  $H_m(x^*)=0$, which implies from \eqref{eq:wH0} that $H_n(x^*)=0$ too.
\end{proof}

The proof of Proposition \ref{prop:easy-Ves} follows directly from Lemma \ref{lem:H} and \ref{lem:simple}.

Although we do not have a proof that the sequence of Hermite polynomials is semi-degenerate, assuming this as a conjecture allows to apply Theorem \ref{thm:symm} and determine the number of real zeros of the Wronskian of an arbitrary sequence of Hermite polynomials.

In the rest of this section we assume that Theorem  \ref{thm:symm} holds for Hermite polynomials, which allows us to proof the following conjecture formulated by Felder et al. in \cite{Felder2012a}:
\begin{conj}\label{conj:Felder}[Felder-Hemery-Veselov]
For doubled partitions $\lambda=(\mu_1^2, \dots, \mu_n^2)$, $H_\lambda$ has no real roots and has as many imaginary roots as there are odd numbers in the partition.
\end{conj}

We assume that $\{H_n\}_{n=0}^\infty$ is semi-degenerate, so that Theorem \ref{thm:symm} holds for Hermite sequences. We shall use the following shorthand notation
\begin{equation}\label{eq:notation}
\lambda= (\mu_1^{m_1}, \dots, \mu_n ^{m_n})=(\underbrace{\mu_1, \dots, \mu_1}_{m_1 \text{ times}}, \dots \underbrace{\mu_n, \dots, \mu_n}_{m_n \text{ times}}).
\end{equation}
We denote by $\bar\lambda$ the conjugate partition to $\lambda$, whose Young diagram is the transpose of the diagram of $\lambda$. We have the following duality property \cite{Felder2012a,GGM13}
\begin{equation} \label{eq:duality}
H_{\bar \lambda}(x)=(-{\rm i})^{|\lambda|}H_{\lambda}({\rm i}x).
\end{equation}
For a doubled partition $\lambda=(\mu_1^2, \dots, \mu_n^2)$, its conjugate partition is given by
\begin{equation} \label{eq:conjpart}
\bar \lambda=(2^{\mu_n-\mu_{n-1}}, 4^{\mu_{n-1}-\mu_{n-2}}, \dots, (2n-2)^{\mu_2-\mu_1}, (2n)^{\mu_1}).
\end{equation}
\begin{proof}[Proof of Conjecture \ref{conj:Felder}]

It is clear that the Wronskian of a doubled partition of Hermite polynomials has no real zeros. This is a particular case of Theorem \ref{thm:symm}, but in fact it is the case treated by Krein \cite{Krein1957} and Adler \cite{Adler}.
In order to compute the number of imaginary zeros, we use the duality property \eqref{eq:duality} to compute the number of real zeros of $H_{\bar\lambda}$, where $\bar\lambda$ is given by \eqref{eq:conjpart}.

Let us first consider the case where $\mu_1,\dots\mu_n$ are odd positive integers. It follows from \eqref{eq:conjpart},\eqref{eq:lambdadef} and \eqref{eq:ddef} that $d_{\bar\lambda}=-1$, and therefore the second term in \eqref{eq:n+} is zero and $H_{\bar\lambda}$ does not vanish at zero. Applying \eqref{eq:n+} to the partition \eqref{eq:conjpart} we see that $n(H_{\bar\lambda})=2n$, as conjectured by Felder et al.
To conclude the proof, we shall see that the number of real zeros does not change when the doubled partition is allowed to contain even integers. 

Let $\nu$ be an even integer and $\lambda$ a doubled partition. Upon the transformation
\[ \lambda=(\mu_1^2, \dots, \mu_n^2)\rightarrow  \lambda'=(\mu_1^2, \dots, \mu_i^2, \nu^2, \mu_{i+1}^2, \dots, \mu_n^2)\]
the conjugate partitions transform as
\[\bar \lambda=(2^{\mu_n-\mu_{n-1}}, \dots, (2n)^{\mu_1})\rightarrow\bar \lambda'=(2^{\mu_n-\mu_{n-1}}, \dots, (2(n-i))^{\mu_{i+1}-\nu},(2(n-i+1))^{\nu-\mu_i}, \dots,  (2n+2)^{\mu_1})\]

Since $d_{\bar \lambda'}=d_{\bar \lambda}$ the difference between the number of real roots of  $H_{\bar \lambda}$ and $H_{\bar \lambda'}$ is
\begin{equation}
n(H_{\bar \lambda'})-n(H_{\bar \lambda})= \sum_{j=1}^{\mu_n}(-1)^{\mu_n-j}(\bar \lambda'_j- \bar \lambda _j)=2\sum_{j=\mu_n-\nu}^{\mu_n} (-1)^{\mu_n-j}=0.
\end{equation}
Thus, the number of pure imaginary roots of $H_\lambda$ for an arbitrary doubled partition is  equal to twice the number of odd numbers in the sequence $\mu_1, \dots, \mu_n$.

\end{proof}

As mentioned in the Introduction, we have proved that the alternate sum formula \eqref{eq:altersum} (or its symmetric variant \eqref{eq:n+}) counts the number of zeros of Wronskians of classical orthogonal polynomials in their interval of orthogonality.
The derivation of this result makes explicit use of the second order differential equation satisfied by these functions, and thus \textit{a priori} there is no reason why it should also apply to the Wronskian of an arbitrary sequence of orthogonal polynomials.
However, numerical evidence seems to suggest that this is indeed the case (see Conjecture \ref{conj:OParb} in Section 1). To prove such a result for orthogonal polynomials with respect to an arbitrary measure would require a different technique, and it would complete the full generalization of Karlin and Szeg\H{o} result for consecutive sequences.
On the other hand, the validity of the alternate sum formula \eqref{eq:altersum} holds for arbitrary eigenfunctions of Schr\"odinger's equation, not just polynomials.

\section*{Acknowledgements}

The authors would like to thank Robert Milson and Antonio Dur\'an for stimulating discussions. The elegant proof of Lemma \ref{lem:H} that uses the irreducibility of Hermite polynomials is in fact entirely due to Robert Milson. MAGF would like to thank the Department of Theoretical Physics II at Universidad Complutense for providing her with office space and all facilities. The research of DGU has been supported in part by the Spanish MINECO-FEDER Grants MTM2012-31714 and FIS2012-38949-C03-01.

\end{document}